\documentclass[12pt]{amsart}
\usepackage{amsmath, amsthm, amssymb}
\newtheorem{theorem}{Theorem}
\newtheorem{prop}{Proposition}

\newcommand{\ep}{\varepsilon}

\begin{document}

\title{Generalization of Doob's Inequality and  A Tighter Estimate on Lookback Option Price}
\author{Jian Sun\\Fudan University}

\begin{abstract}
In this short note, we will  strengthen the classic Doob's $L^p$
inequality for sub-martingale processes.  Because this inequality
is of fundamental importance to the theory of stochastic process,
we believe  this generalization will find many interesting
applications.
\end{abstract}

\thanks{The author would like to thank Professor Steven Shreve
at Carnegie Mellon University and Professor Jiangang Ying at Fudan University. The author also would like to thank
Peter Carr and Bruno Dupire for their inspirational discussions.}

\maketitle

\section{Introduction}

Doob's maximum inequality for the sub-martingale process has
played an important role in the stochastic process theory. It has
become a standard result which appears in almost every
introductory text in this subject. Let $\{X_{t}\}$ be be a process
defined on a probability space with a filtration
$\mathcal{F}_{t}$. Its maximum is defined by

\begin{align*}
M_{t}&=\max _{0\le s\le t} X_{s}\\
\end{align*}
By the definition we have $M_{0}=X_{0}$. If $X_{t}$ is a positive
continuous sub-martingale, the Doob's $L^p$ inequality states that
\begin{equation}
\|M_{T}\|_{p} \le \frac{p}{p-1} \|X_{T}\|_{p}.
\end{equation}
In particular, when $p=2$, we have
\begin{equation}\label{doob.1}
\|M_{T}\| _{2} \le 2 \|X_{T}\| _{2}.
\end{equation}
Even though the coefficient on the right hand side  is not important for the purpose
of establishing the finiteness of the $L^2$ integrability of
$M_{T}$, it may become important for some other applications. For
example, when $X_{T}$ is a martingale with $X_{0}=0$, we
infer from~\eqref{doob.1}
\begin{equation}
E(M_{T})\le E(M_{T}^2)^{\frac{1}{2}} \le 2 E(X_{T}^2)^{\frac{1}{2}}.
\end{equation}
This provides an estimate on the expectation of the Maximum
$M_{T}$ in terms of the standard deviation of $X_{T}$. In fact, we
will see later that we can have a much tighter estimate
\begin{equation}
E(M_{T})\le  E(X_{T}^2)^{\frac{1}{2}},
\end{equation}
for any continuous martingale with $X_{0}=0$. When $X_{0}\neq 0$
 we will have
\begin{equation}
E(M_{T})\le \sqrt{2}E(X_{T}^2)^{\frac{1}{2}}.
\end{equation}
In either case, we  obtain a  better result. In Finance,
people usually use martingales to model the stock prices or any
other tradable assets. Their maximum $M_{T}$ sometimes represent
payoff of certain derivatives. The inequality above actually gives
a good estimate of this derivative payoff in terms of European
type option prices. In this area, the magnitude of the coefficient
matters a lot to the applications.

In general when $p>2$, Doob's inequality is equivalent to
\begin{equation}
E(M_{T}^p) \le \left(\frac{p}{p-1}\right)^{p} E(X_{T}^p),
\end{equation}
and we are going to strengthen this inequality to
\begin{equation}
E(M_{T}^{p})+\frac{p}{p-1}X_{0}^{p}\le \left(
\frac{p}{p-1}\right)^{p} E(X_{T}^{p})
\end{equation}
by adding a term of initial position $X_{0}$.

Our method is to first prove an identity which will involve
$X_{t}$ and $M_{t}$. From this identity we will use the standard
methods to derive our inequalities. Because our starting point is
an identity rather than an estimates, we could prove a tighter
inequality. Our methodology also provides a totally different
proof of Doob's maximum inequality.

\section{Classic Results}

For completeness and comparison, we state and prove the classic
maximal inequality in this section.

\begin{theorem}[Doob's Inequality ]
Let $X_{t}$ be an nonnegative martingale process. For any real
$a>0$ we have
\begin{equation}\label{ineq:doob.3}
a P\left( M_{T}\ge a\right) \le E\left(X_{T}1_{M_{T}\ge a}\right)
\end{equation}
\end{theorem}
\begin{proof}
 We define the  stopping time
\begin{equation}
\tau =\inf \{t: X_{t}\ge a\}.
\end{equation}
and we claim to have
\begin{equation}
1_{\tau \le T} +\frac{X_{T}-X_{T\wedge \tau}}{a}\le
\frac{X_{T}1_{M_{T}\ge a}}{a}.
\end{equation}
It is obvious to check its validity. Take the expectation and use
the fact that $X_{T\wedge\tau}$ is also a martingale, we get the
result.
\end{proof}

\begin{theorem}[Maximal inequality]
For the nonnegative martingale process,  We have the following
$L^{p}$ norm inequality: for any $p>1$,
\begin{equation}\label{ineq:doob.maximum.1}
\|M_{T}\|_{p} \le \frac{p}{p-1} \|X_{T}\|_{p}
\end{equation}
\end{theorem}
\begin{proof}[Classical proof]
This is based on the Doob's inequality. Use the standard measure
theory and the Doob's inequality,
\begin{align*}
E\left(M_{T}^{p}\right) &=\int_{0}^{\infty} p \,x^{p-1}
P\left(M_{T}>x\right)\, dx\\
&\le \int _{0}^{\infty} p \,x^{p-2} E\left(X_{T}1_{M_{T}\ge
x}\right)\, dx\\
&=E\left(\int _{0}^{\infty} p\,x^{p-2}X_{T}1_{M_{T}\ge x}
\,dx\right)\\
&=E\left(\int_{0}^{M_{T}}p\,x^{p-2}X_{T}\,dx\right)\\
&=\frac{p}{p-1} E\left(M_{T}^{p-1}X_{T}\right)\\
&\le \frac{p}{p-1}
E\left(M_{T}^{p}\right)^{(p-1)/p}E\left(X_{T}^{p}\right)^{1/p}
\end{align*}
Further simplify this inequality we will get the result. Please
note that we have used H\"{o}lder inequality
\begin{equation}
E(M_{T}^{p-1}X_{T})\le
E(M_{T}^{p})^{\frac{p-1}{p}}E(X_{T}^{p})^{\frac{1}{p}}
\end{equation}
in this proof.
\end{proof}

\section{A New Inequality}
We will try to strengthen the maximal inequality proved in the
previous section. First we prove an identity.
\begin{theorem}\label{th:identity.1}
Let $X_{t}$ be a nonnegative continuous martingale. For any $p>0$,
if
\begin{equation}
\int _{0}^{T} M_{t}^{2p}d[X,X]_t < \infty,
\end{equation} we have the following identity:
\begin{equation}\label{eq:jian.main}
E\left(X_{T}M_{T}^{p}\right)=\frac{p}{p+1}E(M_{T}^{p+1})+\frac{1}{p+1}X_{0}^{p+1}.
\end{equation}
\end{theorem}
\begin{proof}
We consider the following differential identity:
\begin{equation}
d(X_{t}M_{t}^{p})=M_{t}^{p}dX_{t}+pM_{t}^{p-1}X_{t}dM_{t}.
\end{equation}
Written in the integral term and make use the observation
\begin{equation*}
dM_{t}\neq 0 \Rightarrow X_{t}=M_{t},
\end{equation*}
we have
\begin{equation}
X_{T}M_{T}^{p}-X_{0}^{p+1} =\int_{0}^{T} M_{t}^{p}dX_{t}+\int
_{0}^{T} pM_{t}^{p}dM_{t}.
\end{equation}
Take the expectation and use the fact that $X_{t}$ is a
martingale, therefore
\begin{equation}
E\left(\int_{0}^{T} M_{t}^{p}dX_{t}\right)=0, \end{equation} we
have
\begin{equation*}
E\left(X_{T}M_{T}^{p}\right)=\frac{p}{p+1}E(M_{T}^{p+1})+\frac{1}{p+1}X_{0}^{p+1}.
\end{equation*}
This finishes the proof.
\end{proof}

\begin{theorem}\label{th:jian.main.p=2}
Let $X_{t}$ be a  continuous martingale, if
\begin{equation}
\int _{0}^{T} M_{t}^{2}d[X,X]_t < \infty,
\end{equation} we have the following identity:
\begin{equation}
E\left(X_{T}M_{T}\right)=\frac{1}{2}E(M_{T}^{2})+\frac{1}{2}X_{0}^{2}.
\end{equation}
\end{theorem}
\begin{proof}
The proof is the same as above since when $p=1$, we don't require
$X_{t}$ to be positive anymore.
\end{proof}

\begin{theorem}
Let $X_{t}$ be a nonnegative continuous sub-martingale. For any
$p>0$, if
\begin{equation}
\int _{0}^{T} M_{t}^{2p}d[X,X]_t < \infty,
\end{equation}
we have the following inequality:
\begin{equation}\label{ineq:jian8}
E\left(X_{T}M_{T}^{p}\right)\ge
\frac{p}{p+1}E(M_{T}^{p+1})+\frac{1}{p+1}X_{0}^{p+1}.
\end{equation}
\end{theorem}
\begin{proof}
We basically follow the proof in Theorem~\eqref{th:identity.1}. We
notice that when $X_{t}$ is a sub-martingale,
\begin{equation}
\int_{0}^{T} M_{t}^{p}dX_{t}
\end{equation}
is also a sub-martingale so
\begin{equation}
E\left(\int_{0}^{T} M_{t}^{p}dX_{t}\right)\ge 0
\end{equation}
and this will finish the proof.
\end{proof}

\begin{theorem}[Generalization of Doob's maximal inequality]
For a nonnegative continuous sub-martingale process, if
\begin{equation}
\int _{0}^{T} M_{t}^{2p}d[X,X]_t < \infty,
\end{equation}
we then have
\begin{equation}\label{ineq:jian.doob}
E(M_{T}^{p+1})+\frac{p+1}{p}X_{0}^{p+1}\le \left(
\frac{p+1}{p}\right)^{p+1} E(X_{T}^{p+1})
\end{equation}
for any $p>0$.
\end{theorem}
\begin{proof}
We use the H\"{o}lder inequality: \begin{equation}
E(M_{T}^{p}X_{T})\le
E(X_{T}^{p+1})^{\frac{1}{p+1}}\,E(M_{T}^{p+1})^{\frac{p}{p+1}}.
\end{equation}
For any $0< \ep <1$, we can write
\begin{equation}
E(M_{T}^{p}X_{T})\le \left(\ep ^{-p}
E(X_{T}^{p+1})\right)^{\frac{1}{p+1}}\,\left(\ep
E(M_{T}^{p+1})\right)^{\frac{p}{p+1}}.
\end{equation}
Again, we use the H\"{o}lder inequality
\begin{equation}
a^{\frac{1}{p+1}}b^{\frac{p}{p+1}}\le
\frac{1}{p+1}a+\frac{p}{p+1}b
\end{equation}
to get
\begin{equation}
E(M_{T}^{p}X_{T})\le \frac{\ep ^{-p}}{p+1} E(X_{T}^{p+1}) +
\frac{p\,\ep}{p+1} E(M_{T}^{p+1}).
\end{equation}
Now use the inequality~\eqref{ineq:jian8},
\begin{equation*}
\frac{p}{p+1}E(M_{T}^{p+1})+\frac{1}{p+1}X_{0}^{p+1} \le \frac{\ep
^{-p}}{p+1} E(X_{T}^{p+1}) + \frac{p \, \ep}{p+1} E(M_{T}^{p+1}).
\end{equation*}
Rearranging the terms,
\begin{equation}\label{ineq:jian.4}
E(M_{T}^{p+1}) +\frac{1}{p(1-\ep)}X_{0}^{p+1} \le \frac{p+1}{p}
\frac{1}{\ep ^{p}(1-\ep)(p+1)}E(X_{T}^{p+1}).
\end{equation}
Now we minimize the function
\begin{equation}
\min _{0<\ep <1} \frac{1}{\ep ^{p}(1-\ep)(p+1)}
=\left(\frac{p+1}{p}\right)^{p},
\end{equation}
the equality takes place when $\ep=p/(p+1)$. Put everything back
into (\ref{ineq:jian.4}), we get
\begin{equation*}
E(M_{T}^{p+1})+\frac{p+1}{p}X_{0}^{p+1}\le \left(
\frac{p+1}{p}\right)^{p+1} E(X_{T}^{p+1}).
\end{equation*}
\end{proof}

\section{Some Implications}
\begin{prop}
For $X_{T}$ is a continuous Martingale, then we have
\begin{equation}
E(M_{T}^2)+2X_{0}^2 \le 4E(X_{T}^2).
\end{equation}
\end{prop}
\begin{proof}
For the martingale $X_{t}$, by Theorem~\ref{th:jian.main.p=2},
\begin{align*}
E(M_{T}^{2})+X_{0}^{2}& =2E\left(X_{T}M_{T}\right)\\
&\le \frac{1}{2}E(M_{T}^2)+2E(X_{T}^2)
\end{align*}
and arranging terms will prove the inequality.
\end{proof}

It is interesting to compare with the the classical result which
only gives
\begin{equation}
E(M_{T}^2) \le 4E(X_{T}^2).
\end{equation}
If we use again H\"{o}lder inequality, we will get
\begin{equation}
E(M_{T})^2 \le E(M_{T}^2) \le 4E(X_{T}^2).
\end{equation}
and consequently, we have
\begin{equation}
E(M_{T})\le 2 E(X_{T}^2)^{\frac{1}{2}}.
\end{equation}

We can in fact get stronger result by using the
Identity~\ref{eq:jian.main}.
\begin{prop}
For  continuous martingale $X_{T}$, we have
\begin{equation}
E(M_{T})\le \sqrt {2}E(X_{T}^2)^{\frac{1}{2}}.
\end{equation}
\end{prop}
\begin{proof}
Let $p=1$ in the identity~\ref{eq:jian.main}, we have
\begin{equation*}
E(M_{T}^2)+X_{0}^2 =2E(X_{T}M_{T})
\end{equation*}
which is equivalent to
\begin{equation}\label{eq:jian7}
E((M_{T}-X_{T})^2)=E(X_{T}^2)-X_{0}^2.
\end{equation}
Use H\"{o}lder inequality, we have
\begin{equation}\label{ineq:jian5}
E(M_{T})-E(X_{T}) \le \sqrt{E(X_{T}^2)-X_{0}^2},
\end{equation}
hence,
\begin{equation}
E(M_{T}) \le X_{0}+ \sqrt{E(X_{T}^2)-X_{0}^2}.
\end{equation}
Now use the inequality $(a+b)^2 \le 2(a^2+b^2)$, we have
\begin{equation}
E(M_{T})^2  \le 2E(X_{T}^2).
\end{equation}
which is what we want.
\end{proof}

\begin{prop}
When $X_{T}$ is a  martingale and $X_{0}=0$, then we have
\begin{equation}\label{ineq:jian6}
E(M_{T})\le E(X_{T}^2)^{\frac{1}{2}}.
\end{equation}
\end{prop}
\begin{proof}
In the Inequality ~\ref{ineq:jian5}, take $X_{0}=0$. It is
evident.
\end{proof}


\begin{thebibliography}{99}

\bibitem {} Ioannis Karatzas, Steven  Shreve
 {\it Brownian Motion and Stochastic Calculus}, Springer Verlag, 1991.
\bibitem {} Steven Shreve {\it Stochastic Calculus Models for Finance} Springer Verlag,
2002.
\bibitem {} Daniel Revuz, Marc Yor {\it Continuous Martingales and Brownian Motion} Springer Verlag, 2002.


\end{thebibliography}
\end{document}